\def\stoc{0}
\def\notes{0}

\documentclass[11pt]{article}
\usepackage{amsthm}
\usepackage{amsfonts,amsmath}
\usepackage{amssymb}
\usepackage{xspace}
\usepackage[vmargin=1in,hmargin=1in]{geometry}
\usepackage[dvips,pdfstartview=FitH,pdfpagemode=None,colorlinks=true,citecolor=cyan,linkcolor=blue]{hyperref}
\usepackage{color}

\theoremstyle{plain}
\newtheorem*{theorem*}{Theorem}
\newtheorem*{lem*}{Lemma}

\theoremstyle{plain}
\newtheorem{thm}{Theorem}[section]
  \theoremstyle{definition}
  \newtheorem{defn}[thm]{Definition}
  \theoremstyle{plain}
  
  \theoremstyle{plain}
  \newtheorem{conjecture}[thm]{Conjecture}
  \theoremstyle{plain}
  \newtheorem{lem}[thm]{Lemma}
  \theoremstyle{remark}
  \newtheorem{rem}[thm]{Remark}
    \theoremstyle{plain}
  
    \theoremstyle{plain}
  \newtheorem{claim}[thm]{Claim}
    \theoremstyle{plain}
  \newtheorem{proposition}[thm]{Proposition}

\newcommand{\F}{{\mathbb{F}}}
\newcommand{\E}{{\mathbb{E}}}

\newcommand{\R}{{\mathbb{R}}}
\newcommand{\inp}[1]{\langle{#1}\rangle}

\newcommand{\spn}[1]{{\rm{span}}\left(#1\right)}

\newcommand{\set}[1]{\left\{{#1}\right\}}

\newcommand{\eps}{\epsilon}

\newcommand{\poly}{{\rm{poly}}}
\newcommand{\spec}{{\rm{Spec}}}

\newcommand{\cc}{{\rm{CC}}}
\newcommand{\rnk}{{\rm{rank}}}
\newcommand{\rep}{{\rm{rep}}}
\newcommand{\qpoly}{{\rm{q.poly}}}

\newcommand{\enote}[1]{\ifnum\notes=1{{\sf\color{green} [Eli: #1]}}\fi}
\newcommand{\nnote}[1]{\ifnum\notes=1{{\sf\color{red} [Noga: #1]}}\fi}
\newcommand{\snote}[1]{\ifnum\notes=1{{\sf\color{blue} [Shachar: #1]}}\fi}

\newcommand{\restatethm}[3]{\medskip \noindent {\bf #1} (#2). \textsl{#3}\\}

\begin{document}

\ifnum\stoc=1
\begin{titlepage}
\fi

\title{
An additive combinatorics approach to the log-rank conjecture in communication complexity}
\author{Eli Ben-Sasson\thanks{Department of Computer Science, Technion, Haifa, Israel and Microsoft Research New-England, Cambridge, MA.
 {\tt eli@cs.technion.ac.il}. The research leading to these results has received funding from the European Community's Seventh Framework
Programme (FP7/2007-2013) under grant agreement number 240258.} \and Shachar Lovett\thanks{School of Mathematics, Institute for Advanced Study, Princeton, NJ. {\tt slovett@math.ias.edu}. Supported by NSF grant DMS-0835373.}
 \and Noga Zewi \thanks{Department of Computer Science, Technion, Haifa. nogaz@cs.technion.ac.il}}
\maketitle
\enote{Added ``in communication complexity'' to title}

\begin{abstract}
\enote{Changed title+2nd paragraph of abstract, objections?}
  For a $\{0,1\}$-valued matrix $M$ let $\cc(M)$ denote the deterministic communication complexity of the boolean function associated with $M$.
  The log-rank conjecture of Lov\'{a}sz and Saks [FOCS 1988]
  states that $\cc(M)\leq \log^c(\rnk(M))$ for some absolute constant $c$ where $\rnk(M)$
  denotes the rank of $M$ over the field of real numbers.
 We show that $\cc(M)\leq c \cdot \rnk(M)/\log \rnk(M)$ for some absolute constant $c$, assuming a well-known conjecture from additive combinatorics known as the Polynomial Freiman-Ruzsa (PFR) conjecture.

 Our proof is based on the study of the ``approximate duality conjecture'' which was recently suggested by Ben-Sasson and Zewi [STOC 2011] and studied there in connection to the PFR conjecture. First we improve the bounds on approximate duality assuming the PFR conjecture. Then we use the approximate duality conjecture (with improved bounds) to get the aforementioned upper bound on the communication complexity of low-rank martices, where this part uses the methodology suggested by Nisan and Wigderson [Combinatorica 1995].
\end{abstract}

\ifnum\stoc=1
\end{titlepage}
\setcounter{page}{1}
\fi

\section{Introduction}

This paper presents a new connection between communication complexity and additive combinatorics, showing that a well-known conjecture from additive combinatorics known as the {\em Polynomial Freiman-Ruzsa Conjecture} (PFR, in short), implies better upper bounds than currently known on the deterministic communication complexity of a boolean function in terms of the rank of its associated matrix.
More precisely, our results show that the PFR Conjecture implies that every boolean function has communication
 complexity $O( \rnk(M)/\log \rnk(M))$ where $\rnk(M)$ is the rank, over the reals, of the associated matrix. This is a weakening of a well-known conjecture in communication complexity, known as the {\em log-rank conjecture}, which postulates that every boolean function has communication complexity $\log^{O(1)}\rnk(M)$.

Our analysis relies on the study of {\em approximate duality}, a concept closely related to the PFR Conjecture, which was introduced in \cite{BZ11}. Our main technical contribution is improved bounds on approximate duality, assuming the PFR Conjecture. We then use these bounds in order to prove the new upper bounds on communication complexity.

\subsection{Communication complexity and the log-rank conjecture}

In the two-party communication complexity model two parties --- Alice and Bob --- wish to compute a function $f: X \times Y \to \{0,1\}$ on inputs $x$ and $y$ where $x$ is known only to Alice and $y$ is known only to Bob. In order to compute the function $f$ they must exchange bits of information between each other according to some (deterministic) protocol. The (deterministic) communication complexity of a protocol is the maximum total number of bits sent between the two parties, where the maximum is taken over all pairs of inputs $x,y$. We henceforth omit the adjective ``deterministic'' from our discourse because our results deal only with the deterministic model. The communication complexity of the function $f$, denoted by $\cc(f)$, is the minimum communication complexity of a protocol for $f$.

For many applications it is convenient to associate the function $f: X \times Y \to \{0,1\}$ with the matrix
 $M\in \{0,1\}^{X \times Y}$ whose
 $(x,y)$ entry equals $f(x,y)$. For a $\{0,1\}$-valued matrix $M$, let $\cc(M)$ denote the communication complexity of the boolean function associated with $M$.
 Let $\rnk(M)$ denote the rank of $M$ over the reals. We will occasionally consider the rank of $M$ over the two-element field $\F_2$ and will denote this by $\rnk_{\F_2}(M)$.

It is well-known that $\log \rnk(M) \leq \cc(M) \leq \rnk(M) $ (cf. \cite[Chapter 1.4]{KN97}) and it is a fundamental question to find out what is the true dependency of $\cc(M)$ on the rank. The log-rank conjecture due to\ Lov\'{a}sz and Saks \cite{LS88} postulates that the gap between the above lower bound and upper bound on $\cc(M)$ is polynomial rather than exponential:

\begin{conjecture}[Log-rank]\label{conj:log-rank}
For every $\{0,1\}$-valued matrix M \[\cc(M) = \log^{O(1)} \rnk(M).\]
\end{conjecture}

Lov\'{a}sz and Saks also point out that the above conjecture has several other interesting equivalent formulations. One of them, due to Nuffelen \cite{Nuf76} and Fajtlowicz \cite{Faj88}, is the following:

\begin{conjecture}
For every graph $G$, $\chi(\overline G) \leq \log^{O(1)}\rnk(G)$, where $\chi(\overline G)$ is the chromatic number of the complement of $G$, and $\rnk(G)$ is the rank of the adjacency matrix of $G$ over the reals.
\end{conjecture}

Though considerable effort has been made in an attempt to resolve this conjecture since its introduction 25 years ago, to this date not much is known regarding to it. The best currently known upper bound on the communication complexity of a $\{0,1\}$-valued matrix $M$ with respect to its rank is  $\cc(M) \leq \log (4/3) \cdot \rnk(M) = (0.415...) \cdot \rnk(M)$
due to Kotlov \cite{Kot97} and it improves on the previous best bound of $\cc(M) \leq \rnk(M)/2$ by Kotlov and Lov\'{a}sz \cite{KL96}. In the other direction, the best lower bound is $\cc(M) \geq \Omega(\log^{t} \rnk(M))$ for $t = \log_3 6=1.631 ...$ due to Kushilevitz (unpublished, cf. \cite{NW95}) which improved on  a previous bound with $t = \log_2 3=1.585...$ due to Nisan and Wigderson \cite{NW95}.

Our main result is an upper bound of  $O (\rnk(M)/\log \rnk(M))$ on the communication complexity of a $\{0,1\}$-valued matrix $M$ assuming a well-known conjecture from additive combinatorics --- the Polynomial Freiman-Ruzsa conjecture --- discussed next.

\subsection{Additive combinatorics and the Polynomial Freiman-Ruzsa conjecture}

Quoting the (current) Wikipedia definition, additive combinatorics studies ``combinatorial estimates associated with the arithmetic operations of addition and subtraction''. As such, it deals with a variety of problems that aim to 'quantify' the amount of additive structure in subsets of additive groups.
One such a problem is that which is addressed by the Polynomial Freiman-Ruzsa conjecture (we shall encounter a different problem in additive combinatorics when we get to ``approximate duality'' later on).

For $A \subseteq \F_2^n$, let $A+A$ denote the sum-set of $A$
\[A+A:=\{a+a' \mid a,a'\in A\}\]
where addition is over $\F_2$. It is easy to see that $|A+A|  = |A|$ if and only if $A$ is an affine subspace of $\F_2^n$. The question addressed by the Freiman-Ruzsa Theorem is whether the ratio of $|A+A|$ to $|A|$ also 'approximates' the closeness of $A$ to being a subspace, or in other words, whether the fact that $A+A$ is small with respect to the size of  $A$ also implies that $\spn A$ is small with respect to the size of $A$.
The Freimnan-Ruzsa Theorem \cite{Ruz99} says that this is indeed the case.

\begin{thm}[Freiman-Ruzsa Theorem \cite{Ruz99}]\label{thm:Freiman}
If $A \subseteq \F_2^n$ has $|A+A| \leq K|A|$, then $|\spn A| \leq K^2 2^{K^4}|A|$.
\end{thm}

The above Theorem was improved in a series of works \cite{GR06,San08,GT09} and the best currently known upper bound
on the ratio $\frac {|\spn A|} {|A|}$ is $2^{(2+o(1))K}$
due to Green and Tao \cite{GT09}. This can be seen to be almost tight (up to the $o(1)$ term) by letting $A=\{u_1,u_2,\ldots,u_t\}$, where $u_1,u_2,\ldots,u_t \in \F_2^n$ are linearly independent vectors. Then in this case we have $|A+A| \approx \frac {t} 2 |A|$, while $|\spn A| = 2^t$.

The same example also shows that the ratio $\frac {|\spn A|} {|A|}$ must depend exponentially on $K$.
However, this example does not rule out the existence of a large subset $A' \subseteq A$
for which the ratio $\frac {|\spn{A'}|} {|A'|}$ is just polynomial in $K$, and this is exactly what is suggested by the PFR Conjecture:

\begin{conjecture}[Polynomial Freiman-Ruzsa (PFR)]\label{conj:PFR}
There exists an absolute constant $r$, such that if  $A\subset \F_2^n$ has $|A+A|\leq K|A|$, then there exists a subset
$A' \subseteq A$ of size at least $K^{-r}|A|$ such that $|\spn{A'}| \leq |A|$.
\end{conjecture}

Note that the above conjecture implies that $|\spn{A'}| \leq |A| \leq K^{r}|A'|$.
The PFR conjecture has many other interesting equivalent formulations, see the survey of Green \cite{Green05} for some of them. It is conjectured to hold for subsets of general groups as well and not only for subsets of the group $\F_2^n$ but we will be interested only in the latter case. Significant progress on this conjecture has been achieved recently by Sanders \cite{San10}, using new techniques developed by Croot and Sisask \cite{CS10}. Sanders proved an upper bound on the ratio $\frac {|\spn{A'}|} {|A'|}$ which is quasi-polynomial in $K$:

\begin{thm}[Quasi-polynomial Freiman-Ruzsa Theorem \cite{San10}]\label{thm:Sanders}
Let $A\subset \F_2^n$ be a set such that $|A+A|\leq K|A|$. Then there exists a subset
$A' \subseteq A$ of size at least $K^{-O(\log^{3}K)}|A|$ such that $|\spn {A'}| \leq |A|$.
\end{thm}

 Our main result is the following:

\begin{thm}[Main.]\label{thm:main}
Assuming the PFR Conjecture \ref{conj:PFR}, for every $\{0,1\}$-valued matrix M \[\cc(M) = O(\rnk(M)/\log \rnk(M)).\]
\end{thm}

We end this section by mentioning two other recent applications of the PFR Conjecture to theoretical computer science. The first application is to the area of
low-degree testing, and was discovered by Lovett \cite{Lov10} and, independently, by Green and Tao \cite{GT10}. The second application is to the construction of two-source extractors due to Ben-Sasson and Zewi \cite{BZ11}. The latter paper also introduced the notion of approximate duality which plays a central role in our proof method so we describe it now.

\subsection{Approximate duality}

Our main technical contribution is improved bounds on approximate duality, assuming the PFR conjecture. These new bounds lie at the heart of our proof of the Main Theorem \ref{thm:main}.

The notion of approximate duality was first introduced in \cite{BZ11}. For $A,B \subseteq \F_2^n$, we define the {\em duality measure} of $A, B$ in (\ref{eq:duality-measure-introduction}) as an estimate of how `close' this pair is to being dual

\begin{equation}\label{eq:duality-measure-introduction}
  D(A,B) := \bigg|\E_{a\in A, b\in B}\left[(-1)^{\inp{a,b}_2}\right]\bigg|,
\end{equation}
 where $\inp{a,b}_2$ denotes the binary inner-product of $a,b$ over $\F_2$, defined by $\inp{a,b}_2=\sum_{i=1}^{n} a_i\cdot b_i$ where all arithmetic operations are
 in $\F_2$.

 \begin{rem}
 The duality measure can be alternatively defined as the discrepancy of the inner product function on the rectangle $A \times B$
 (up to a normalization factor of $ \frac {2^n} {|A||B|}$). Nevertheless we chose to use the term 'duality measure' instead of 'discrepancy' because
 of the algebraic context in which we use it, as explained below.
 \end{rem}

It can be verified that if $D(A,B)=1$ then $A$ is contained in an affine shift of $B^\perp$ which is the space dual to the linear $\F_2$-span of $B$. The question is what can be said about the structure of $A,B$ when $D(A,B)$ is sufficiently large, but strictly smaller than $1$. The following theorem from \cite{BZ11} says that if the duality measure is a constant very close to 1 (though strictly smaller than $1$) then there exist relatively large subsets $A' \subseteq A$, $B' \subseteq B$, such that $D(A',B')=1$.

\begin{thm}[Approximate duality for nearly-dual sets, \cite{BZ11}]\label{thm:duality-thm}
  For every $\delta>0$ there exists a constant $\epsilon>0$ that depends only on $\delta$, such that if $A,B\subseteq\F_2^n$ satisfy $D(A,B)\geq 1-\epsilon$, then there exist subsets $A'\subseteq A, |A'|\geq \frac14 |A|$ and $B'\subseteq B, |B'|\geq 2^{-\delta n}|B|$, such that $D(A',B')=1$.
\end{thm}

It is conjectured that a similar result holds also when the duality measure is relatively small, and in particular when it tends to zero as $n$ goes to infinity. Furthermore, the following theorem from \cite{BZ11} gives support to this conjecture, by showing that such bounds indeed follow from the PFR conjecture.

\begin{thm}[Approximate duality assuming PFR, \cite{BZ11}]\label{thm:AD-exp-loss} Assuming the PFR Conjecture \ref{conj:PFR}, for every pair of constants $\alpha,\delta>0$ there exists a constant $\zeta>0$, depending only on $\alpha$ and $\delta$, such that the following holds. If $A,B\subseteq \F_2^n$ satisfy $|A|,|B|>2^{\alpha n}$ and $D(A,B) \geq 2^{-\zeta n}$, then there exist subsets
  $A'\subseteq A$,  $|A'|\geq 2^{-\delta n}|A|$ and  $B'\subseteq B$,  $|B'|\geq 2^{-\delta n}|B|$ such that $D(A',B')=1$.
\end{thm}

Our main technical contribution is an improved bound on approximate duality assuming the PFR conjecture.

\begin{lem}[Main technical lemma]\label{lem:main-technical}
Suppose that $A,B \subseteq \F_2^n$ satisfy $D(A,B) \geq 2^{-\sqrt n}$. Then assuming the PFR conjecture \ref{conj:PFR}, there exist subsets $A',B'$ of $A,B$ respectively such that $D(A',B')=1$, and  $|A'|\geq 2^{-cn/ \log n}|A|$, $|B'|\geq 2^{-cn/ \log n}|B|$ for some absolute constant $c$.
\end{lem}

\begin{rem}
The bound on $\min\left\{\frac{|A'|}{|A|},\frac{|B'|}{|B|}\right\}$ --- which in the lemma above is $2^{-cn/\log n}$ --- cannot be improved beyond $2^{-O(\sqrt{n})}$ even if we assume $D(A,B)>0.99$. To see this take $A=B={n\choose c'\sqrt{n}}$ to be the set of all $\{0,1\}$-vectors with exactly $c'\sqrt{n}$ ones, where $c'$ is a sufficiently small positive constant that guarantees $D(A,B)\geq 0.99$. It can be verified that if $A'\subset A,B'\subset B$ satisfy $D(A',B')=1$ then the smaller set is of size $2^{-\Omega(\sqrt{n})}\cdot |A|$.
\end{rem}

The proof of Lemma~\ref{lem:main-technical} appears in Section \ref{sec:improved AD}. Note that this lemma improves on Theorem \ref{thm:AD-exp-loss} in two ways. First, in the lemma the ratios $|A'|/|A|$, $|B'|/|B|$ are bounded from below by $2^{-cn/ \log n}$, whereas in Theorem  \ref{thm:AD-exp-loss} we only get a smaller bound of the form $2^{-\delta n}$ for some constant $\delta > 0$.
 Note however that this improvement comes with a requirement that the duality measure $D(A,B)$ is larger --- in the above lemma we require that it is at least $2^{-\sqrt n}$ while in Theorem \ref{thm:AD-exp-loss} we only require it to be at least $2^{-\zeta n} \ll 2^{-\sqrt n}$. We note however that the bound $D(A,B) \geq 2^{-\sqrt{n}}$ can be replaced by $D(A,B) \ge \exp(-n^{1-\eps})$
 for any $\eps>0$ at the price of a larger constant $c=c(\eps)$.
 The second improvement is that the parameters in the lemma do not depend on the sizes of the sets $A$ and $B$ whereas in
 Theorem \ref{thm:AD-exp-loss} they do depend on $|A|,|B|$.

These two improvements are precisely what enables us to prove the new upper bound of $O(\rnk(M) / \log \rnk(M))$ on the communication complexity of $\{0,1\}$-valued matrices assuming the PFR conjecture, and we go on to sketch our proof method next.

\subsection{Proof overview}

First we show how our Main Theorem \ref{thm:main} is deduced from the improved bounds on approximate duality in Lemma \ref{lem:main-technical}. Then we give an overview of the proof of Lemma \ref{lem:main-technical} itself.

\enote{Changed the overview to something more informal.}

\paragraph{From approximate duality to communication complexity upper bounds.}
We follow the approach of Nisan and Wigderson from \cite{NW95}. Let the {\em size} of a matrix $M$ be the number of entries in it and if $M$ is $\{0,1\}$-valued let $\delta(M)$ denote its {\em (normalized) discrepancy}, defined as the absolute value of the difference between the fraction of zero-entries  and one-entries in $M$. Informally, discrepancy measures how ``unbalanced'' is $M$, with $\delta(M)=1$ when $M$ is {\em monochromatic} --- all entries have the same value --- and $\delta(M)=0$ when $M$ is completely balanced.

Returning to the work of \cite{NW95}, they observed that to prove the log-rank conjecture it suffices to show that a $\{0,1\}$-valued matrix $M$ of rank $r$ always contains a monochromatic sub-matrix of size $|M|/\qpoly(r)$ where $\qpoly(r)=r^{\log^{O(1)} r}$ means quasi-polynomial in $r$. Additionally, they used spectral techniques (i.e., arguing about the eigenvectors and eigenvalues of $M$) to show that any $\{0,1\}$-valued matrix $M$ of rank $r$ contains a relatively large submatrix $M'$ ---  of size at least $|M|/r^{3/2}$ --- that is somewhat biased --- its discrepancy is at least $1/r^{3/2}$. We show, using tools from additive combinatorics, that $M'$ in fact contains a pretty large monochromatic submatrix (though not large enough to deduce the log-rank conjecture).

To this end we start by working over the two-element field $\F_2$. This seems a bit counter-intuitive because the log-rank conjecture is false over $\F_2$. The canonical counterexample is the inner product function $IP(x,y)= \inp{x,y}_2$ --- It is well-known (see e.g. \cite{KN97}[Chapters 1.3., 2.5.]) that $\rnk_{\F_2} (M_{IP}) = n$ while $\cc(IP) = n$. However, rather than studying $M$ over $\F_2$ we focus on the biased submatrix $M'$ and things change dramatically. (As a sanity-check notice that $M_{IP}$ does not contain large biased submatrices and this does not contradict the work of \cite{NW95} because the rank of $M_{IP}$ over the reals is $2^{n}-1$.)

Thus, our starting point is a large submatrix $M'$ that has large discrepancy. It is well-known that $\rnk_{\F_2}(M')\leq \rnk(M')\leq r$ and that this implies $M'$ can be written as $M'=A^\top \cdot B$ where $A,B$ are matrices whose columns are vectors in $\F_2^{r}$. Viewing each of $A,B$ as the set of its columns, we have in hand two sets that have a large duality measure as defined in (\ref{eq:duality-measure-introduction}), namely,  $D(A,B)= \delta(M') \geq 1/r^{3/2}$. This is setting in which we apply our Main Technical Lemma~\ref{lem:main-technical} and deduce that $A,B$ contain relatively large subsets $A',B'$ with $D(A',B')=1$. One can now verify that the submatrix of $M'$ whose rows and columns are indexed by $A',B'$ respectively is indeed monochromatic, as needed. We point out that to get our bounds we need to be able to find monochromatic submatrices of $M'$ even when $M'$ is both small and skewed (i.e., has many more columns than rows or vice versa). Fortunately, Lemma~\ref{lem:main-technical} is robust enough to use in such settings.

\paragraph{Improved bounds on approximate duality assuming PFR.}

We briefly sketch the proof of our Main Technical Lemma \ref{lem:main-technical}. We use the {\em spectrum} of a set as defined in \cite[Chapter 4]{TV06}:
\begin{defn}[Spectrum]\label{def:spectrum}
  For a set $B \subseteq \F_2^n$ and $\alpha\in[0,1]$ let the {\em $\alpha$-spectrum} of $B$ be the set
\begin{equation}\label{eq:spec}
\spec_\alpha (B): = \{ x \in \F_2^n \;|\;    \mid  \E_{b\in B}\big[(-1)^{\inp{x,b}_2}\big]\mid \geq \alpha \}.
\end{equation}
\end{defn}

Notice that $A \subseteq \spec_{\epsilon}(B)$ implies $D(A,B) \geq \epsilon$ (cf. (\ref{eq:duality-measure-introduction})). In the other direction, Markov's inequality can be used to deduce that $D(A,B) \geq \epsilon$ implies the existence of $A' \subseteq A$ of relatively large size ---
$|A'| \geq \frac \eps 2 |A|$ --- such that $A' \subseteq \spec_{\epsilon/2}(B)$. To prove our lemma we start with $A_1=A'$ and establish a sequence of sets \[A_2 \subseteq A_1+A_1, \ \  A_3 \subseteq A_2+A_2,\ldots\] such that $A_i \subseteq \spec_{\epsilon_i}(B)$ for all $i$. This holds by construction for $A_1$ with $\epsilon_1 = \epsilon/2$, and we show that it is maintained throughout the sequence for increasingly smaller values of $\epsilon_i$ (we shall use $\epsilon_i=\epsilon_{i-1}^2$).

Moving our problem from the field of real numbers to the two-element field $\F_2$ now pays off. Each $A_i$ is of size at most $2^n$ so there must be an index $i \leq n / \log K$ for which  $|A_{i+1}| \leq K|A_{i}|$, let $t$ be the minimal such index. We use the PFR conjecture together with the Balog--Szemer{\'e}di--Gowers Theorem \ref{thm:BSG} from additive combinatorics to show that our assumption that $|A_{t+1}| \leq K|A_{t}|$ implies that a large subset $A''_t$ of $A_t$ has small span (over $\F_2$).

We now have in hand a set $A''_t$ which is a relatively large fraction of its span and additionally satisfies $D(A''_t,B)  \geq \epsilon_t$ because by construction $A''_t\subseteq \spec_{\epsilon_t}(B)$. We use an approximate duality claim from \cite{BZ11} (Lemma~\ref{lem:fourier}) which applies when one of the sets is a large fraction of its span (in our case the set which is a large fraction of its span is $A''_t$). This claim says that $A''_t$ and $B$ each contain relatively large subsets $A'_t, B'_t$ satisfying $D(A'_t,B'_t)=1$. Finally, recalling $A'_t$ is a (carefully chosen) subset of $A_{t-1}+A_{t-1}$, we argue that $A_{t-1}$ contains a relatively large subset $A'_{t-1}$ that is ``dual'' to a large subset $B'_{t-1}$ of $B'_t$, where by ``dual'' we mean $D(A'_{t-1},B'_{t-1})=1$ (in other words $A'_{t-1}$ is contained in an affine shift of the space dual to $\spn{B'_{t-1}}$). \enote{Added formal explanation of ``dual''.} We continue in this manner to find pairs of ``dual'' subsets for $t-2,t-3,\ldots,1$ at which point we have found a pair of ``dual'' subsets of $A,B$ that have relatively large size, thereby completing the proof.

\paragraph{Paper organization.}  The next section contains the proof of the Main Technical Lemma~\ref{lem:main-technical}. The proof of Main Theorem \ref{thm:main} given the Main Technical Lemma \ref{lem:main-technical}
\ifnum\stoc=1
is deferred to Appendix~\ref{sec:upper-bounds-CC} due to space limitations.
\else
appears in Section \ref{sec:upper-bounds-CC}.
\fi
\enote{Changed order of proofs because want the more important part to be contained in first 10 pages of submission.}

\section{Improved bounds on approximate duality assuming PFR}\label{sec:improved AD}
\enote{Would be nice to give more structure to this section. After all, this is our main technical contribution. }

In this section we prove our Main Technical Lemma \ref{lem:main-technical} by proving the following more general version of it.

\begin{lem}[Main technical lemma, general form]\label{lem:improved ADC}
Suppose that $A,B \subseteq \{0,1\}^n$ satisfy $D(A,B) \geq \epsilon$. Then assuming the PFR Conjecture \ref{conj:PFR}, for every $K \ge 1$ and $t = n/\log K$, there exist subsets $A',B'$ of $A,B$ respectively such that $D(A',B')=1$, and
\begin{equation}\label{eq:222}|A'| \geq \poly \bigg( \frac {(\epsilon/2)^{2^{t}}} {nK} \bigg) (4n)^{-t} |A|,
\quad \quad |B'| \geq \poly \bigg(\frac {(\epsilon/2)^{2^t}} {nK}\bigg) 2^{-t} |B|.\end{equation}
\end{lem}

\begin{proof}[Proof of Lemma~\ref{lem:main-technical}]
Follows from Lemma~\ref{lem:improved ADC} by setting
$K=2^{4   n/ \log n}$, $t= \frac{\log n}{4}$, $\epsilon =2^{-\sqrt n}$.
\end{proof}

\paragraph{Additive combinatorics preliminaries} In what follows all arithmetic operations are taken over $\F_2$. For the proof of Lemma~\ref{lem:improved ADC} we need two other theorems from additive combinatorics. The first is the well-known Balog--Szemer{\'e}di--Gowers Theorem of \cite{BalogS,Gowers}.

\begin{thm}[Balog--Szemer{\'e}di--Gowers]\label{thm:BSG}
There exist fixed polynomials $f(x,y)$, $g(x,y)$ such that the following holds for every
subset $A$ of an abelian additive group. If $A$ satisfies $\Pr_{a,a' \in A}[a+a'\in S] \ge 1/K$ for $|S| \leq C |A|$, then one can find a subset $A' \subseteq A$ such that $|A'| \geq |A|/f(K,C)$, and $|A'+A'| \leq g(K,C)|A|$.
\end{thm}

The second is a lemma from \cite{BZ11} which can be seen as an approximate duality statement which applies when one of the sets has small span:
\enote{Changed statement of theorem. Previously had a tiny bug, said $|A'|\geq \frac12 |A|$}
\begin{lem}[Approximate-duality for sets with small span, \cite{BZ11}]\label{lem:fourier}
If $D(A,B) \geq \epsilon$, then there exist subsets $A' \subseteq A, B' \subseteq B$, $|A'| \geq \frac{\epsilon}{4} |A|$,
$|B'| \geq \frac{\epsilon^2}{4} \frac{|A|} {|\spn A|}|B|$,
such that $D(A',B')=1$. If $A\subseteq \spec_\epsilon(B)$ then we have $|A'|\geq |A|/2$ and $|B'|\geq \epsilon^2\frac{|A|} {|\spn A|}|B|$ in the statement above.
\end{lem}

%
Recall the definition of the spectrum given in (\ref{eq:spec}):
$$
\spec_\alpha (B): = \{ x \in \F_2^n \;|\;    \mid  \E_{b\in B}\big[(-1)^{\inp{x,b}_2}\big]\mid \geq \alpha \}.
$$
Finally, for $S\subset\F_2^n$ and $x\in \F_2^n$ let $\rep_S(x)$ be the number of different representations of $x$
as an element of the form $s+s'$ where $s,s' \in S$. $\rep_S(x)$ can also be written, up to a normalization factor, as $1_{S} * 1_{S}(x)$ where $1_S$ is the indicating function of the set $S$ and $*$ denotes convolution.

\paragraph{Proof overview}
We construct a decreasing sequence of constants \[\eps_1=\eps/2, \eps_2=\eps_1^2/2, \eps_3=\eps_2^2/2,\ldots\] and a sequence of sets \[A_1:=A \cap \spec_{\epsilon_1}(B) ,\quad  A_2 \subseteq (A_1+A_1)\cap\spec_{\eps_2}(B), \quad A_3 \subseteq (A_2+A_2)\cap \spec_{\eps_3}(B), \ldots \]
Since each of the sets in the sequence is of size at most $2^n$ there must be an index $i \leq n / \log K$ for which  \begin{equation}
  \label{eq:pf-overview} |A_{i+1}| \leq K|A_{i}|
\end{equation} and let $t$ be the minimal such index.
The PFR Conjecture~\ref{conj:PFR} together with the Balog--Szemer{\'e}di--Gowers Theorem~\ref{thm:BSG} will be used to deduce from (\ref{eq:pf-overview}) that a large subset $A''_t$ of $A_t$ has  small span. Applying Lemma \ref{lem:fourier} to the sets $A''_t$ and $B$ implies the existence of large subsets $A'_t \subseteq A_t$ and $B'_t \subseteq B$ such that $D(A'_t,B'_t)=1$. Finally we argue inductively for $i=t-1,t-2,\ldots,1$ that there exist large subsets $A'_i \subseteq A_i$ and $B'_i \subseteq B$ such that $D(A'_i,B'_i)=1$. The desired conclusion will follow from the $i=1$ case. To be able to ``pull back'' and construct a pair of large sets $A'_{i-1}, B'_{i-1}$ from the pair $A'_{i}, B'_{i}$ we make sure every element in $A_i$ is the sum of roughly the same number of pairs in $A_{i-1} \times A_{i-1}$.

\paragraph{The sequence of sets}
\enote{Moved definition of $\rep$ to beginning of proof.}
Let $\epsilon_1:=\epsilon/2$, $A_1:=A \cap \spec_{\epsilon_1}(B)$. Assuming $A_{i-1},\eps_{i-1}$ have been defined set $\eps_i=\eps_{i-1}^2/2$ and let $j_i\in\set{0,\ldots, n-1}$ be an integer index which maximizes the size of
\begin{equation}\label{eq:size_A_i}
\left\{(a,a') \in A_{i-1}\; |\; a+a' \in \spec_{\epsilon_i}(B) \textrm{ and } 2^{j_i} \le \rep_{A_{i-1}}(a+a') \le 2^{j_i+1}\right\}.
\end{equation}
and set
\begin{equation}\label{eq:def_A_i}
A_i := \{a+a': a,a' \in A_{i-1}, a+a' \in \spec_{\epsilon_i}(B) \textrm{ and } 2^{j_i} \le \rep_{A_{i-1}}(a+a') \le 2^{j_i+1}\}.
\end{equation}

\begin{claim}\label{clm:111} For $i=1$ we have $|A_1| \ge (\epsilon/2) |A|$. For $i>1$ we have
\begin{equation}\label{eq:prob_sum_in_A_i}
\Pr_{a,a' \in A_{i-1}}[a+a' \in A_i] \ge \eps_i/n
\end{equation}
and additionally
\begin{equation}\label{eq:lower_bound_size_A_i}
|A_i| \ge \frac{\eps_i}{2^{j_i+1} n}|A_{i-1}|^2.
\end{equation}
\end{claim}

\begin{proof}
The case of $i=1$ follows directly from Markov's inequality. For larger $i$ we argue that
$$
\Pr_{a,a' \in A_{i-1}}[a+a' \in \spec_{\epsilon_{i}}(B)] \ge \epsilon_{i}.
$$
To see this use Cauchy-Schwarz to get
$$
\E_{a,a' \in A_{i-1}} |\E_{b \in B} (-1)^{\inp{a+a',b}}| =  E_{b\in B}(\E_{a \in A_{i-1}}[ (-1)^{\inp{a,b}}])^2 \ge (\E_{a \in A_{i-1},b \in B}[ (-1)^{\inp{a,b}}])^2 = \eps_{i-1}^2
$$
and apply Markov's inequality to deduce that an $\eps_i$-fraction of $(a,a')\in A_{i-1}\times A_{i-1}$ sum to an element of $\spec_{\eps_i}(B)$. Selecting $j_i$ to maximize (\ref{eq:size_A_i}) yields inequality (\ref{eq:prob_sum_in_A_i}). Since every element $x \in A_i$ can be represented as $x=a+a'$ with $a,a' \in A_{i-1}$ in at most $2^{j_i+1}$ different ways we deduce (\ref{eq:lower_bound_size_A_i}) from (\ref{eq:prob_sum_in_A_i}) and complete the proof.
\end{proof}

\paragraph{The inductive claim}
Let $t$ be the minimal index such that $|A_{t+1}| \leq K |A_t|$ and note that $t \leq n / \log K$ because all sets $A_i$ are contained in $\F_2^n$. We shall prove the following claim by backward induction.

\begin{claim}[Inductive claim]\label{clm:inductive}
For $i=t,t-1,\ldots,1$ there exist subsets \[A'_i \subseteq A_i, \quad B'_i \subseteq B_i\] such that $D(A'_i,B'_i)=1$ and $A'_i,B'_i$ are not too small:
\[|A'_i| \geq \poly \bigg(\frac{\epsilon_{t+1}} {nK}\bigg) (4n)^{-(t-i)} \bigg( \prod_{\ell =i}^{t} \epsilon_{\ell+1} \bigg) |A_i|, \quad \quad
|B'_i| \geq \poly \bigg(\frac{\epsilon_{t+1}} {nK}\bigg) 2^{-(t-i)} |B|\]
\end{claim}

We split the proof of the claim to two parts. The base case (Proposition~\ref{lem:induc base}) is proved using the tools from additive combinatorics listed in the beginning of this section. The inductive step is proved in Proposition~\ref{lem:induc step} using a graph construction. Before proving Claim~\ref{clm:inductive} we show how it implies Lemma~\ref{lem:improved ADC}.

\begin{proof}
  [Proof of Main Technical Lemma~\ref{lem:improved ADC}]
  Set $i=1$ in Claim~\ref{clm:inductive} above.
  Recall that $\epsilon_{i+1} = \epsilon_i^2/2$ for all $i$, so
\[\epsilon_{\ell+1} = \epsilon^{2^{\ell}} / 2^{2^{\ell}-1} \ge  (\epsilon/2)^{2^{\ell}}.\] Thus we have
$\epsilon_{t+1} \geq (\epsilon/2)^{2^t}$ and
$\prod_{\ell=1}^t \epsilon_{\ell+1} \ge (\epsilon/2)^{2^{t+1}}.$ This gives the bounds on $A',B'$ stated in (\ref{eq:222}).
\end{proof}

\begin{proposition}[Base case of Claim~\ref{clm:inductive}  ($i=t$)]\label{lem:induc base}
There exist subsets $A'_t \subseteq A_t$, $B'_t \subseteq B_t$ such that $D(A'_t,B'_t)=1$ and $A'_t,B'_t$ are not too small:
\[|A'_t| \geq \poly \bigg(\frac{\epsilon_{t+1}} {nK}\bigg)  |A_t|, \quad \quad |B'_t| \geq \poly \bigg(\frac{\epsilon_{t+1}} {nK}\bigg)  |B|.\]
\end{proposition}

\begin{proof}

By assumption $|A_{t+1}| \leq K |A_t|$ and $Pr_{a,a' \in  A_t}[a+a'\in A_{t+1}] \geq \epsilon_{t+1}/n$ by (\ref{eq:lower_bound_size_A_i}). Hence we can apply the
Balog--Szemer{\'e}di--Gowers Theorem (Theorem \ref{thm:BSG}) to the set $A_t$ to obtain
a subset $\tilde A_t \subseteq A_t$ such that
\[ |\tilde A_t| \geq \poly \bigg(\frac{\epsilon_{t+1}} {nK} \bigg) |A_t| ,\] and
 \[ |\tilde A_t + \tilde A_t| \leq \poly \bigg(\frac{nK} {\epsilon_{t+1}} \bigg) |A_t|  = \poly \bigg(\frac{nK} {\epsilon_{t+1}}\bigg) |\tilde A_t|.\]

Now we can apply the PFR Conjecture~\ref{conj:PFR} to the set $\tilde A_t$ which gives a subset $A''_t \subseteq \tilde A_t$ such that
\[ | A''_t| \geq \poly \bigg(\frac{\epsilon_{t+1}} {nK} \bigg) |\tilde A_t| = \poly \bigg(\frac{\epsilon_{t+1}} {nK} \bigg) |A_t| ,\] and
\[ |\spn {A''_t}| \leq |\tilde A_t| = \poly \bigg(\frac {nK} {\epsilon_{t+1}} \bigg)|A''_t| .\]

Recall that $A''_t \subseteq \spec_{\epsilon_{t}}(B)$, and in particular $D(A''_t,B) \geq \epsilon_{t}$.
Applying Lemma \ref{lem:fourier} to the sets $A''_t$ and $B$ we conclude that there exist subsets $A'_t \subseteq A''_t$, $B' \subseteq B$ such that $D(A'_t,B')=1$, and which satisfy $|A'_t| \geq \frac 12 |A''_t|$ and
\[|B'_t| \geq \epsilon_{t}^2 \frac {|A''_t|}{ |\spn {A''_t}|} |B| =  \poly \bigg(\frac{\epsilon_{t+1}} {nK}\bigg)|B|.\]
This completes the proof of the base case.
\end{proof}

%
%
%
%
%
%

\begin{proposition}[Inductive step of Claim~\ref{clm:inductive}]\label{lem:induc step}
For every $i=t-1,\ldots,1$ there exist subsets $A'_i \subseteq A_i$, $B'_i \subseteq B_i$ such that $D(A'_i,B'_i)=1$ and $A'_i,B'_i$ are not too small:
\[|A'_i| \geq \poly \bigg(\frac{\epsilon_{t+1}} {nK}\bigg) (4n)^{-(t-i)} \bigg( \prod_{\ell =i}^{t} \epsilon_{\ell+1} \bigg) |A_i|, \quad  \quad
|B'_i| \geq \poly \bigg(\frac{\epsilon_{t+1}} {nK}\bigg) 2^{-(t-i)} |B|.\]
\end{proposition}

\begin{proof} Suppose that the claim is true for $i$ argue it holds for index $i-1$. Let $G=(A_{i-1},E)$ be the graph whose vertices are the elements in $A_{i-1}$, and $(a,a')$ is an edge if $a+a' \in A'_i$.
We bound the number of edges in this graph from below.
Recall from (\ref{eq:def_A_i}) that every $a \in A'_i$ (where $A'_i\subseteq A_i$) satisfies $2^{j_i} \leq \rep_{A_{i-1}}(a) \leq 2^{j_i+1}$. Using this we get

$$
\begin{array}{llll}
|E| & \geq & 2^{j_i}\cdot |A'_i|  &\text{($\rep_{A_{i-1}}(x)\geq 2^{j_i}$ for all $x\in A'_i$)}\\
&  \geq & 2^{j_i}\cdot \poly \bigg(\frac{\epsilon_{t+1}} {nK}\bigg) (4n)^{-(t-i)} \bigg( \prod_{\ell =i}^{t} \epsilon_{\ell+1} \bigg) |A_i|  & \text{(induction hypothesis)}\\
&\geq & 2^{j_i} \cdot \poly \bigg(\frac{\epsilon_{t+1}} {nK}\bigg) (4n)^{-(t-i)} \bigg( \prod_{\ell =i}^{t} \epsilon_{\ell+1} \bigg) \frac{\epsilon_i}{2^{j_i+1}n} |A_{i-1}|^2  & \text{(by~\eqref{eq:lower_bound_size_A_i})}\\
&=  & 2 \cdot \poly \bigg(\frac{\epsilon_{t+1}} {nK}\bigg) (4n)^{-(t-(i-1))} \bigg( \prod_{\ell =i-1}^{t} \epsilon_{\ell+1} \bigg) |A_{i-1}|^2.  &
\end{array}
$$

Let $M:=\poly \bigg(\frac{\epsilon_{t+1}} {nK}\bigg) (4n)^{-(t-(i-1))} \bigg( \prod_{\ell =i-1}^{t} \epsilon_{\ell+1} \bigg)$. Since our graph has at least $2M|A_{i-1}|^2$ edges and
$|A_{i-1}|$ vertices, it has a connected component with at least $2M|A_{i-1}|$ vertices and denote by $A''_{i-1}$ the set of vertices in it.

Choose an arbitrary element $a$ in $A''_{i-1}$.
Partition $B'_i$ into two sets $B'_{i,0}$ and $B'_{i,1}$ such that all elements in $B'_{i,0}$ have inner product $0$ with $a$, and all elements in $B'_{i,1}$
have inner product $1$ with $a$. Let $B'_{i-1}$ be the larger of $B'_{i,0}$,$B'_{i,1}$, and note that $|B'_{i-1}| \geq |B'_i|/2$. Recall that our assumption was
that $D(A'_i,B'_i)=1$. Abusing notation, let $\inp{A'_{i},B'_{i}}_2$ denote the value of $\inp{a',b'}_2$ for some $a'\in A'_i, B'_i$ (the choice of $a',b'$ doesn't matter because $D(A'_i,B'_i)=1$).
Next we consider two cases --- the case where $\inp{A'_{i},B'_{i}}_2=0$, and the case where $\inp{A'_{i},B'_{i}}_2=1$.

In the first case we have that for every $a,a' \in A''_{i-1}$ which are neighbors in the graph, $a+a' \in A'_{i}$, and therefore $\inp{a+a',b}_2=0$ for every $b\in B'_{i-1}$.
This implies in turn that $\inp{a,b}_2= \inp{a',b}_2$ for all elements $a,a'\in A''_{i-1}$ which are neighbors in the graph, $b\in B'_{i-1}$. Since $A''_{i-1}$  induces a connected component, and  due to our
choice of $B'_{i-1}$, this implies that $D(A''_{i-1},B'_{i-1})=1$ so we set $A'_{i-1}=A''_{i-1}$.

In the second case we have that $\inp{a+a',b}_2=1$ for every  $a,a' \in A''_{i-1}$ which are neighbors in the graph, $b\in B'_{i-1}$. In particular this implies that $\inp{a,b}_2= \inp{a',b}_2+1$ for every elements $a,a'\in A''_{i-1}$ which are neighbors in the graph, $b\in B'_{i-1}$. This means that $A''_{i-1}$ can be partitioned into two sets $A'_{i-1,0},A'_{i-1,1}$, where the first one contains all elements in $A''_{i-1}$ that have inner product $0$ with all elements in $B'_{i-1}$, while the second set contains all elements in $A'_{i-1}$ that have inner product $1$ with all elements in $B'_{i-1}$. We set $A'_{i-1}$ to be the larger of these two sets and get $D(A'_{i-1},B'_{i-1})=1$ and $|A'_{i-1}| \geq M|A_{i-1}|$.

Concluding, in both cases we obtained subsets $A'_{i-1},B'_{i-1}$ of $A_{i-1},B$ respectively, such that $D(A'_{i-1},B'_{i-1})=1$ and $A'_{i-1}, B'_{i-1}$ are not too small:
\[|A'_{i-1}| \geq \poly \bigg(\frac{\epsilon_{t+1}} {nK}\bigg) (4n)^{-(t-(i-1))} \bigg( \prod_{\ell =i-1}^{t} \epsilon_{\ell+1} \bigg)  |A_{i-1}| ,\]
and
\[|B'_{i-1}| \geq \frac 12 |B'_i| \geq \frac 12 \poly \bigg(\frac{\epsilon_{t+1}} {nK}\bigg) 2^{-(t-i)} |B|= \poly \bigg(\frac{\epsilon_{t+1}} {nK}\bigg) 2^{-(t-(i-1))} |B|.\]
This concludes the proof of the inductive claim.

\end{proof}

\ifnum\stoc=1

\paragraph{Acknowledgements.} We thank Nati Linial and Eyal Kushilevitz for drawing our attention, independently, to the similarities between the notions of discrepancy and approximate duality, which led us to consider the question addressed in this paper.
\newpage

\bibliography{log-rank}
\bibliographystyle{alpha}   
\newpage
\appendix
\fi

\section{From approximate duality to communication complexity upper bounds}\label{sec:upper-bounds-CC}

In this section we prove our main theorem, Theorem \ref{thm:main} given the main technical lemma, Lemma \ref{lem:main-technical}.
The proof of the main technical lemma is deferred to Section \ref{sec:improved AD}.

We start by repeating the necessary definitions. For a $\{0,1\}$-valued matrix $M$,
let $\cc(M)$ denote the communication complexity of the boolean function associated with $M$.
Let $\rnk(M)$ and  $\rnk_{\F_2}(M)$ denote the rank of $M$ over the reals and over $\F_2$, respectively. We denote  by $|M|$ the
total number of entries in $M$, and by
$|M_0|$ and $|M_1|$ the number of zero and non-zero entries of $M$, respectively.
 We say that $M$ is {\em monochromatic} if  either $|M|=|M_0|$ or $|M|=|M_1|$. Finally,
we define the {\em discrepancy} $\delta(M)$ of $M$ to be the ratio $\frac {||M_0|-|M_1||} {|M|}$.

Recall the statements of Theorem \ref{thm:main} and Lemma~\ref{lem:main-technical}.

\restatethm{Theorem \ref{thm:main} - Main theorem}{restated}{
Assuming the PFR conjecture (Conjecture \ref{conj:PFR}), for every $\{0,1\}$-valued matrix M, \[\cc(M) = O(\rnk(M)/\log \rnk(M)).\]
}
\restatethm{Lemma \ref{lem:main-technical} - Main technical lemma}{restated}{
Suppose that $A,B \subseteq \F_2^n$ satisfy $D(A,B) \geq 2^{-\sqrt n}$. Then assuming the PFR conjecture, there exist subsets $A',B'$ of $A,B$ respectively
such that $D(A',B')=1$, and  $|A'|\geq 2^{-cn/ \log n}|A|$, $|B'|\geq 2^{-cn/ \log n}|B|$ for some absolute constant $c$.
}

We first prove that the above lemma is equivalent to the following one:

\begin{lem}[Main technical lemma, equivalent matrix form]\label{lem:main-technical matrix form}
Let $M$ be a $\{0,1\}$-valued matrix with no identical rows or columns, of rank at most $r$ \underline{over $\F_2$}, and of discrepancy at least $2^{-\sqrt r}$. Then assuming the PFR conjecture
(Conjecture \ref{conj:PFR}),  there exists a monochromatic submatrix $M'$ of $M$ of size at least $2^{-c r / \log r}|M|$ for some absolute constant $c$.
\end{lem}

\begin{proof}
We prove only the Lemma \ref{lem:main-technical} $\Rightarrow$ Lemma \ref{lem:main-technical matrix form} implication. The proof of the converse implication is similar.
Denote the number of rows and columns of $M$ by $k,\ell$ respectively.  It is well known that the rank of $M$ over a field $\F$ equals $r$ if and only if
 $M$ can be written as the sum of $r$ rank one matrices over the field $\F$.
Since $\rnk_{\F_2}(M) \leq r$ this implies in turn that
there exist subsets $A,B \subseteq \F_2^r$, $A=\{a_1,a_2,\ldots,a_k\}$, $B=\{b_1,b_2,\ldots,b_{\ell}\}$ such that $M_{i,j}= \inp{a_i,b_j}_2$
for all $1 \leq i\leq k, 1 \leq j \leq \ell$. Since $M$ has no identical rows or columns we know that $|A|=k$, $|B|=\ell$.
Note that $D(A,B)=\delta(M) \geq 2^{-\sqrt r}$.

Lemma \ref{lem:main-technical} now implies the existence of subsets $A' \subseteq A$, $B' \subseteq B$, $|A'|\geq 2^{-cr/ \log r}|A|$,
$|B'|\geq 2^{-cr/ \log r}|B|$, such that $D(A',B')=1$. Let $M'$ be the submatrix of $M$ whose rows and columns correspond to the indices in $A'$
and $B'$ respectively.
The fact that $D(A',B')=1$ implies that $M_{i,j}=\inp{a_i,b_j}_2 \equiv \rm{const}ã$ for all $a_i\in A'$, $b_j\in B'$. Therefore $M'$ is a monochromatic submatrix of $M$ of
which satisfies \[|M'|=|A'||B'| \geq 2^{-2cr/ \log r}|A||B| = 2^{-2c r/ \log r }|M|,\] as required.
\end{proof}

In order to prove Theorem \ref{thm:main} we follow the high-level approach of Nisan and Wigderson \cite{NW95} which was explained in the previous section. They showed that in order to prove the log-rank conjecture it suffices to prove that every $\{0,1\}$-valued matrix of low rank has a large monochromatic submatrix.
We start with the following lemma.

\begin{lem}[Existence of large monochromatic submatrix assuming PFR]\label{lem: mono-r/logr}
Assuming the PFR conjecture, every $\{0,1\}$-valued matrix $M$ with no identical rows or columns has a monochromatic submatrix of size at least  $2^{-O(\rnk(M) / \log \rnk(M))} |M|$.
\end{lem}

In order to prove the above lemma we use Lemma \ref{lem:main-technical matrix form},
together with the following theorem from \cite{NW95}, which says that every $\{0,1\}$-valued matrix $M$ contains a submatrix of high discrepancy:

\begin{thm}[Existence of submatrix with high discrepancy \cite{NW95}]\label{thm:disc-log-rank}
Every $\{0,1\}$-valued matrix $M$ has a submatrix $M'$ of size at least $(\rnk(M))^{-3/2}|M| $ and with  $\delta(M') \geq (\rnk(M))^{-3/2} $.
\end{thm}

\begin{proof}[Proof of Lemma \ref{lem: mono-r/logr}]
Let $r=\rnk(M)$. Theorem \ref{thm:disc-log-rank} implies the existence of a submatrix $M'$ of $M$ with $|M'| \geq (\rnk(M))^{-3/2}|M| $, and
 $\delta(M') \geq r^{-3/2} \gg 2^{-\sqrt r}$.
Note also that \[\rnk_{\F_2}(M') \leq \rnk(M') \leq \rnk(M)= r.\]

Lemma \ref{lem:main-technical matrix form} then implies the existence of a monochromatic submatrix $M''$ of $M'$ of size at least
 $2^{-c r / \log r}|M'|$ for some absolute constant $c$.
So we have that $M''$ is a monochromatic submatrix of $M$ which satisfies
\[|M''| \geq 2^{-c r / \log r}|M'|  \geq 2^{-c r / \log r} r^{-3/2}|M| = 2^{-O(r/ \log r)} |M|\]

\end{proof}

\begin{proof}[Proof of Theorem \ref{thm:main}]

Let $M$ be a $\{0,1\}$-valued matrix. We will construct a deterministic protocol for $M$ with communication complexity $O(\rnk(M)/\log \rnk(M))$. We may assume w.l.o.g that $M$ has no repeated rows or columns, otherwise we can eliminate the repeated row or column
and the protocol we construct for the ``compressed'' matrix (with no repeated rows/columns) will also be a protocol for $M$.

We follow the high level approach of the proof of Theorem 2 from \cite{NW95}. We will show a protocol with $ 2^{O(r/ \log r)}$ leaves.
This will suffice
since it is well-known  that a protocol with $t$ leaves has communication complexity at most $O(\log t)$ (cf. \cite[Chapter 2, Lemma 2.8]{KN97}). \enote{Would be nice to give reference to the paper that first proved this.} \nnote{According to KN this lemma is folklore and was never published}

Now we describe the protocol.
Let $Q$ be the largest monochromatic submatrix of $M$. Then $Q$ induces a natural partition of $M$ into 4 submatrices $Q,R,S,T$ with $R$ sharing the rows of $Q$ and $S$ sharing the columns of $Q$.
\[M = \left( \begin{array}{ccc}
Q & R  \\
S & T \end{array} \right)\]

Let $U_1$ be a subset of the rows of $(Q|R)$ whose restriction to the columns of $R$ span the rows of $R$. Similarly, let $U_2$ be a subset of the rows of $(S|T)$ whose restriction to the columns of $S$  span the rows of $S$.
Note that if $Q$ is the all zeros matrix then the rows of $U_1$ are independent of the rows of $U_2$. Otherwise, if $Q$ is the all ones matrix then the rows of $U_1$ are independent of all the rows of $U_2$ except possibly for
the vector in $U_2$ whose restriction to the columns of $S$ is the all ones vector (if such vector exists).
Thus since $Q$ is monochromatic we have that
 $\rnk(R) + \rnk(S)  = |U_1| + |U_2|  \leq \rnk(M) +1 $.

\nnote{Will add a drawing of the matrix later. Need to figure out how to do it.}

If $\rnk(R)\leq \rnk(S)$ then the row player sends a bit saying if his input belongs to the rows of $Q$ or not.
The players continue recursively with a protocol for the submatrix $(Q|R)$ or the submatrix $(S|T)$ according to the bit sent.
If $\rnk(R) \geq \rnk(S)$ the roles of the row and column players are switched.

Suppose without loss of generality that $\rnk(R)\leq \rnk(S)$. Then after sending one bit we continue with either the matrix $(Q|R)$ which is of rank at most $\rnk(M)/2$ or with the matrix $(S|T)$ which --- thanks to Lemma \ref{lem: mono-r/logr}  --- is of size at most $(1-\delta)|M|$ for $\delta \geq 2^{- cr/ \log r}$.

Let $L(m,r)$ denote the number of leaves in the protocol starting with a matrix of area at most $m$ and rank at most $r$. Then we get the following recurrence relation:

\[L(m,r) \leq \left\{\begin{array}{ll} L(m,r/2)+L(m(1-\delta),r) & r>1 \\
1 & r=1\end{array}\right.\]

It remains to show that in the above recursion $L(m,r) = 2^{O(r/ \log r)}$. Applying the recurrence iteratively $1/\delta$ times to the right-most summand we get
$$
L(m,r) \le \delta^{-1} L(m,r/2) + L(m(1-\delta)^{1/\delta},r) \le 2^{c r/\log(r)} L(m,r/2) + L(m/2,r).
$$
Set $A(m,r):= 2^{-2cr / \log r}L(m,r)$. Then we have
$A(m,r) \leq A(m,r/2)+A(m/2,r)$ which together with $A(1,r),A(m,1) \leq 1$
imply
$A(m,r) \leq {\log m + \log r \choose \log r}$
since we may apply the recursion iteratively at most $\log r$ times to the left term and $\log m$ times to the right term before we reach $A(1,r)$ or $A(m,1)$.
This in turn implies $A(m,r) \leq {\log m + \log r \choose \log r} \leq r^{O(\log r)}$ due to
the fact that $r \leq  m \leq 2^{2r}$, since we may assume there are no identical rows or columns in the matrix $M$.

Concluding, we have $L(m,r) \leq 2^{2cr/ \log r + O(\log^2 r) }$, which implies in turn $\cc(M) = O(r/ \log r)$ as claimed.
\end{proof}

\ifnum\stoc=0

\paragraph{Acknowledgements.} We thank Nati Linial and Eyal Kushilevitz for drawing our attention, independently, to the similarities between the notions of discrepancy and approximate duality, which led us to consider the question addressed in this paper.

\fi
\end{document}